\documentclass[11pt]{llncs}
\usepackage{amsmath,amssymb,graphicx,mathrsfs}
\usepackage{tabularx,booktabs,multirow,delarray,array}
\usepackage{latexsym}
\usepackage{epstopdf}
\usepackage[linesnumbered, vlined, ruled]{algorithm2e}

\usepackage[in]{fullpage}

\def\calC{\mathcal{C}}
\def\calI{\mathcal{I}}
\def\scrL{\mathscr{L}}

\def\rarrow{\overrightarrow}

\newtheorem{observation}{Observation}

\begin{document}

\title{Dispersing Points on Intervals\thanks{A preliminary version of this paper will appear in the Proceedings of the 27th International Symposium Algorithms and Computation (ISAAC 2016).}
}

\author{
Shimin Li
\and
Haitao Wang
}

\institute{
Department of Computer Science\\
Utah State University, Logan, UT 84322, USA\\
\email{shiminli@aggiemail.usu.edu, haitao.wang@usu.edu}\\
}

\maketitle

\pagestyle{plain}
\pagenumbering{arabic}
\setcounter{page}{1}

\begin{abstract}
We consider a problem of dispersing points on disjoint intervals on
a line. Given $n$ pairwise disjoint intervals sorted on a
line, we want to find a point in each interval such that the
minimum pairwise distance of these points is maximized.
Based on a greedy strategy, we present a linear time algorithm for the problem.
Further, we also solve in linear time the cycle version of the problem where
the intervals are given on a cycle.
\end{abstract}

\section{Introduction}
\label{sec:intro}
The problems of dispersing points have been extensively studied and
can be classified to different categories by their
different constraints and objectives, e.g.,
\cite{ref:FernandezTh13,ref:JagerSo07,ref:ProkopyevTh09,ref:RaviHe94,ref:RaviFa91,ref:WangA88}.
In this paper,
we consider problems of dispersing points on intervals in linear domains
including lines and cycles.

Let $\mathcal{I}$ be a set of $n$ intervals on a line $\ell$, and no
two intervals of $\calI$ intersect. The problem is to find a point in
each interval of $\calI$ such that the minimum distance of any
pair of points is maximized. We assume the intervals of $\calI$ are
given sorted on $\ell$. In this paper
we present an $O(n)$ time algorithm for this problem.

As an application of the problem, consider the following scenario. Suppose we are given $n$ pairwise disjoint intervals on $\ell$ and we want to build a facility on each interval. As the facilities can interfere with each other if they are too close (e.g., if the facilities are hazardous), the goal is to choose locations for these facilities such that the minimum pairwise distance among these facilities is minimized. Clearly, this is an instance of our problem.


We also consider the {\em cycle version} of the problem where the
intervals of $\calI$ are given on a cycle $\calC$. The intervals of
$\calI$ are also pairwise disjoint and are given sorted cyclically on
$\calC$. Note that the distance
of two points on $\calC$ is the length of the shorter arc of $\calC$
between the two points. By making use of our ``line version''
algorithm, we solve this cycle version problem in linear time as well.

\subsection{Related Work}

To the best of our knowledge, we have not found any previous work on
the two problems studied in this paper. Our problems essentially belong to a family of geometric dispersion problems, which are NP-hard in general in two and higher dimensional space.
For example, Baur and Fekete~\cite{ref:BaurAp01} studied the problems of distributing a number of points within a polygonal region such that the points are dispersed far away from each other, and they showed that the problems cannot be approximated arbitrarily well in polynomial time, unless P=NP.

Wang and Kuo \cite{ref:WangA88} considered the following two problems. Given a
set $S$ of points and a value $d$, find a largest subset of $S$ in which the
distance of any two points is at least $d$. Given a
set $S$ of points and an integer $k$, find a subset of $k$ points of $S$ to
maximize the minimum distance of all pairs of points in the subset.
It was shown in \cite{ref:WangA88}
that both problems in 2D are NP-hard but can be
solved efficiently in 1D. Refer to
\cite{ref:BenkertA09,ref:ErkutTh90,ref:FowlerOp81,ref:FurediTh91,ref:MaranasNe95} for other geometric
dispersion problems.
Dispersion problems in various non-geometric settings were also considered
\cite{ref:FernandezTh13,ref:JagerSo07,ref:ProkopyevTh09,ref:RaviFa91,ref:RaviHe94}.
These problems are in general NP-hard; approximation and heuristic algorithms
were proposed for them.

On the other hand, problems on intervals usually have applications in other
areas. For example, some problems on intervals are related to scheduling
because the time period between the release time and the deadline of a
job or task in scheduling problems can be considered as an interval
on the line. From the interval point of view,
Garey et al. \cite{ref:GareySc81} studied the following problem on
intervals: Given $n$ intervals on a line, determine whether it is
possible to find a unit-length sub-interval in each input interval, such
that these sub-intervals do not intersect.
An $O(n\log n)$ time algorithm was given in \cite{ref:GareySc81} for this
problem. The optimization version of the above problem was also studied
\cite{ref:ChrobakOn07,ref:VakhaniaA13}, where the goal is to
find a maximum number of intervals that contain non-intersecting
unit-length sub-intervals. Chrobak et al.~\cite{ref:ChrobakOn07}
gave an $O(n^5)$ time algorithm for the problem, and later Vakhania
\cite{ref:VakhaniaA13} improved the algorithm to $O(n^2\log n)$ time.
The online version of the problem was also considered
\cite{ref:ChrobakA06}.
Other optimization problems on intervals have also been considered, e.g., see
\cite{ref:GareySc81,ref:LangSc76,ref:SimonsA78,ref:VakhaniaMi13}.



\subsection{Our Approaches}
\label{sec:approach}
For the line version of the problem, our algorithm is based on a greedy strategy.
We consider the intervals of $\calI$ incrementally from left to right,
and for each interval, we will ``temporarily'' determine a point in the interval.
During the algorithm, we maintain a value $d_{\min}$, which is the
minimum pairwise distance of the ``temporary'' points that so far have been
computed.  Initially, we put a point at the left endpoint of the first
interval and set $d_{\min}=\infty$. During the algorithm, the value
$d_{\min}$ will be monotonically decreasing. In general, when the next
interval is considered, if it is possible to put a point in the
interval without decreasing $d_{\min}$, then we put such a point
as far left as possible. Otherwise, we put a point on the right endpoint
of the interval. In the latter case, we also need to adjust the points
that have been determined temporarily in the previous intervals that
have been considered. We adjust these points in a greedy way such that
$d_{\min}$ decreases the least. A straightforward implementation of this
approach can only give an $O(n^2)$ time algorithm. In order to achieve
the $O(n)$ time performance, during the algorithm we maintain a
``critical list'' $\scrL$ of intervals, which is a subset of intervals
that have been considered. This list has some properties that
help us implement the algorithm in $O(n)$ time.

We should point out that our algorithm is fairly simple and easy to
implement. In contrast, the rationale of the idea is quite involved and
it is not an easy task to argue its correctness.
Indeed, discovering the critical list is the most challenging work and it
is the key idea for solving the problem in linear time.



To solve the cycle version, the main idea is to
convert the problem to a problem instance on a line and then
apply our line version algorithm.
More specifically,
we make two copies of the intervals of $\calI$ to a line
and then apply our line version algorithm on these $2n$ intervals on the line. The line
version algorithm will find $2n$ points in these intervals and we show that
a particular subset of $n$ consecutive points of them correspond
to an optimal solution for the original problem on $\calC$.

In the following, we will present our algorithms for the line version
in Section \ref{sec:alg.line}. The cycle version is discussed in
Section \ref{sec:cycle}. Section \ref{sec:conclusion} concludes.

\section{The Line Version}
\label{sec:alg.line}

Let $\calI=\{I_1,I_2,\ldots,I_n\}$ be the set of intervals sorted
from left to right on $\ell$. For any two
points of $p$ and $q$ on $\ell$, we use $|pq|$ to denote their
distance. Our goal is to
find a point $p_i$ in $I_i$ for each $1\leq i\leq n$, such that the
minimum pairwise distance of these points, i.e., $\min_{1\leq i< j\leq
n}|p_ip_j|$, is maximized.

For each interval $I_i$, $1\leq i\leq n$, we use $l_i$ and $r_i$ to
denote its left and right endpoints, respectively.
We assume $\ell$ is the $x$-axis. With a little abuse
of notation, for any point $p\in \ell$, depending on the context, $p$
may also refer to its coordinate on $\ell$.
Therefore, for each $1\leq i\leq n$, it is required that $l_i\leq
p_i\leq r_i$.

For simplicity of discussion, we make a general position assumption
that no two endpoints of the intervals of $\calI$ have the same
location (our algorithm can be easily extended to the general case). Note that this implies $l_i<r_i$ for any interval $I_i$.

The rest of this section is organized as follows.
In Section \ref{sec:obser}, we discuss some  observations.
In Section \ref{sec:overview}, we give an overview of our algorithm.
The details of the algorithm are presented in Section
\ref{sec:details}. Finally,
we discuss the correctness and analyze the running time in Section
\ref{sec:correct}.

\subsection{Observations}
\label{sec:obser}


Let $P=\{p_1,p_2,\ldots,p_n\}$ be the set of sought points. Since all intervals
are disjoint, $p_1<p_2<\ldots<p_n$.
Note that the minimum pairwise distance of the points of $P$ is also the minimum
distance of all pairs of adjacent points.

Denote by $d_{opt}$ the minimum pairwise distance of $P$ in an
optimal solution, and $d_{opt}$ is called the {\em optimal objective value}.
We have the following lemma.

\begin{lemma}\label{lem:10}
$d_{opt}\leq\frac{r_j - l_i}{j-i}$ for any $1\leq i<j\leq n$.
\end{lemma}
\begin{proof}
Assume to the contrary that this is not true. Then there exist $i$ and
$j$ with $i<j$ such that $d_{opt}>\frac{r_j - l_i}{j-i}$. Consider any
optimal solution OPT.
Note that in OPT, $p_i, p_{i+1}, \ldots, p_j$ are located in the intervals
$I_i,I_{i+1},\ldots,I_j$, respectively,  and $|p_ip_j|\geq d_{opt}\cdot (j-i)$.
Hence, $|p_ip_j|>r_j-l_i$.  On the other hand, since
$l_i\leq p_i$ and $p_j\leq r_j$, it holds that $|p_ip_j|\leq
r_j - l_i$. We thus obtain contradiction.  \qed
\end{proof}

The preceding lemma leads to the following corollary.

\begin{corollary}\label{coro:10}
Suppose we find a solution (i.e., a way to place the points of $P$) in which the minimum
pairwise distance of $P$ is equal to $\frac{r_j-l_i}{j-i}$ for some
$1\leq i<j\leq n$. Then the solution is an optimal solution.
\end{corollary}

Our algorithm will find such a solution as stated in
the corollary.

\subsection{The Algorithm Overview}
\label{sec:overview}

Our algorithm will consider and process the intervals of $\calI$
one by one from left to right. Whenever an interval $I_i$ is processed, we will
``temporarily'' determine $p_i$ in $I_i$. We say ``temporarily''
because later the algorithm may change the location of $p_i$. During
the algorithm, a value $d_{\min}$ and two indices $i^*$ and $j^*$ will be maintained
such that $d_{\min}=(r_{j^*}-l_{i^*})/(j^*-i^*)$ always holds.

Initially, we set $p_1=l_1$ and $d_{\min}=\infty$, with $i^*=j^*=1$.
In general, suppose the first
$i-1$ intervals have been processed; then $d_{\min}$ is equal to the
minimum pairwise distance of the points $p_1,p_2,\ldots,p_{i-1}$, which have
been temporarily determined. In fact, $d_{\min}$ is the optimal
objective value for the sub-problem on the first $i-1$ intervals.
During the execution of algorithm, $d_{\min}$ will be
monotonically decreasing. After all intervals are processed, $d_{\min}$ is
$d_{opt}$. When we process the next interval $I_i$, we
temporarily determine $p_i$ in a greedy manner as follows.
If $p_{i-1}+d_{\min}\leq l_i$, we put $p_i$ at $l_i$.
If $l_i<p_{i-1}+d_{\min}\leq  r_i$, we put $p_i$ at $p_{i-1}+d_{\min}$.
If $p_{i-1}+d_{\min}>  r_i$, we put $p_i$ at $r_i$.
In the first two cases, $d_{\min}$ does not change.
In the third case, however, $d_{\min}$ will
decrease. Further, in the third case, in order to make the decrease of $d_{\min}$ as
small as possible, we need to move
some points of $\{p_1,p_2,\ldots,p_{i-1}\}$ leftwards. By a straightforward
approach, this moving procedure can be done in $O(n)$ time. But
this will make the entire algorithm run in $O(n^2)$ time.

To have any hope of obtaining an $O(n)$ time algorithm, we need to perform the above
moving ``implicitly'' in $O(1)$ amortized time. To this end, we need
to find a way to answer the following question: Which points of
$p_1,p_2,\ldots,p_{i-1}$ should move leftwards and how far should they move? To
answer the question, the crux of our
algorithm is to maintain a ``critical list'' $\scrL$ of interval indices, which
bears some important properties that eventually help us
implement our algorithm in $O(n)$ time.

In fact, our algorithm is fairly simple. The most
``complicated'' part is to use a linked list to store $\scrL$ so that the following
three operations on $\scrL$ can be performed in constant time each: remove the
front element; remove the rear element; add a new element to the
rear. Refer to Algorithm \ref{algo:line} for the pseudocode.

Although the algorithm is simple, the rationale of the idea is rather
involved and it is also not obvious to see the correctness.
Indeed, discovering the critical list is the most challenging task and the
key idea for designing our linear time
algorithm.
To help in understanding and give some intuition, 
below we use an example of only three intervals to illustrate how the algorithm
works.

Initially, we set $p_1=l_1$, $d_{\min}=\infty$, $i^*=j^*=1$, and $\scrL=\{1\}$.

To process $I_2$, we first try to put $p_2$ at $p_1+d_{\min}$. Clearly,
$p_1+d_{\min}>r_2$. Hence, we put $p_2$ at $r_2$. Since $p_1$ is
already at $l_1$, which is the leftmost point of $I_1$,
we do not need to move it. We update $j^*=2$ and $d_{\min}=r_2-l_1$.
Finally, we add $2$ to the rear of $\scrL$. This finishes the
processing of $I_2$.

Next we process $I_3$.  We try to put $p_3$ at $p_2+d_{\min}$.
Depending on whether
$p_2+d_{\min}$ is to the left of $I_3$, in $I_3$, or to the right of
$I_3$, there are three cases (e.g., see Fig.~\ref{fig:threeintervals}).

\begin{figure}[t]
\begin{minipage}[t]{\textwidth}
\begin{center}
\includegraphics[height=0.6in]{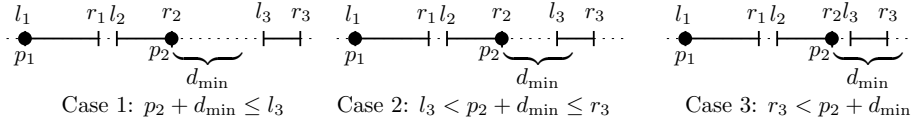}
\caption{\footnotesize Illustrating the three cases when $I_3$ is
being processed.}
\label{fig:threeintervals}
\end{center}
\end{minipage}
\vspace{-0.15in}
\end{figure}

\begin{enumerate}
\item
If $p_2+d_{\min}\leq l_3$, we set $p_3=l_3$. We reset $\scrL$ to
$\{3\}$. None of $d_{\min}$, $i^*$, and $j^*$ needs to be changed in this case.


\item
If $l_3< p_2+d_{\min}\leq r_3$, we set $p_3=p_2+d_{\min}$.
None of $d_{\min}$, $i^*$, and $j^*$ needs to be changed.
Further, the critical list $\scrL$ is updated as follows.

We first give some ``motivation'' on why we need to update $\scrL$.
Assume later in the algorithm, say, when we process the next interval,
we need to move both $p_2$ and $p_3$ leftwards simultaneously so that
$|p_1p_2|=|p_2p_3|$ during the moving (this is for making $d_{\min}$ as large as possible).
The moving procedure stops once either $p_2$ arrives at $l_2$ or $p_3$
arrives at $l_3$. To determine which case happens first, it
suffices to determine whether $l_2-l_1 > \frac{l_3-l_1}{2}$.

\begin{enumerate}
\item
If $l_2-l_1 > \frac{l_3-l_1}{2}$, then $p_2$ will arrive at $l_2$ first, after
which $p_2$ cannot move
leftwards any more in the rest of the algorithm but $p_3$ can still move leftwards.

\item
Otherwise, $p_3$ will arrive at $l_3$ first, after which $p_3$ cannot
move leftwards any more. However, although $p_2$ can still move leftwards,
doing that would not help in making $d_{\min}$ larger.
\end{enumerate}


We therefore update $\scrL$ as follows.
If $l_2-l_1 > \frac{l_3-l_1}{2}$, we add $3$ to the rear of $\scrL$.
Otherwise, we first remove $2$ from the rear of $\scrL$ and then add $3$ to the rear.

\item
If $r_3< p_2+d_{\min}$, we set $p_3=r_3$. Since
$|p_2p_3|<d_{\min}$, $d_{\min}$ needs to be decreased. To make
$d_{\min}$ as large as possible, we will move $p_2$
leftwards until either $|p_1p_2|$ becomes equal to $|p_2p_3|$ or
$p_2$ arrives at $l_2$. To determine which event happens first, we only
need to check whether $l_2-l_1>\frac{r_3-l_1}{2}$.

\begin{enumerate}
\item
If $l_2-l_1>\frac{r_3-l_1}{2}$, the latter event happens first. We set
$p_2=l_2$ and update $d_{\min}=r_3-l_2$ ($=|p_2p_3|$), $i^*=2$, and $j^*=3$. Finally, we
remove $1$ from the front of $\scrL$
and add $3$ to the rear of $\scrL$, after which $\scrL=\{2,3\}$.

\item
Otherwise, the former event happens first. We set
$p_2=l_1+\frac{r_3-l_1}{2}$ and update $d_{\min}=(r_3-l_1)/2$ ($=|p_1p_2|=|p_2p_3|$) and $j^*=3$ ($i^*$ is still
$1$). Finally, we update $\scrL$ in the same way as the above
second case. Namely, if $l_2-l_1>\frac{l_3-l_1}{2}$, we add $3$ to
the rear of $\scrL$; otherwise, we remove $2$ from $\scrL$ and add
$3$ to the rear.
\end{enumerate}

\end{enumerate}

One may verify that in any case the above obtained $d_{\min}$ is an
optimal objective value for the three intervals.

As another example,
Fig.~\ref{fig:figure1} illustrates the solution found by our
algorithm on six intervals.

\begin{figure}[t]
\begin{minipage}[t]{\textwidth}
\begin{center}
\includegraphics[width=\textwidth]{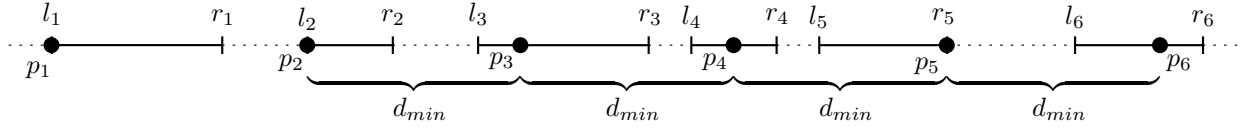}
\caption{\footnotesize Illustrating the solution computed by our
algorithm, with $i^*=2$ and $j^*=5$.}
\label{fig:figure1}
\end{center}
\end{minipage}
\end{figure}

\subsection{The Algorithm}
\label{sec:details}

We are ready to present the details of our algorithm.
For any two indices $i<j$, let $P(i,j)=\{p_i,p_{i+1},\ldots,p_j\}$.

Initially we set
$p_1=l_1$, $d_{\min}=\infty$, $i^*=j^*=1$, and $\scrL=\{1\}$.
Suppose interval $i-1$ has just been processed for some $i>1$.
Let the current critical list be $\scrL=\{k_s,k_{s+1},\ldots k_t\}$ with $1\leq
k_s< k_{s+1}< \cdots < k_t\leq i-1$, i.e., $\scrL$ consists
of $t-s+1$ sorted indices in $[1,i-1]$.
Our algorithm maintains the following {\em invariants}.

\begin{enumerate}
\item
The ``temporary'' location of $p_{i-1}$ is known.

\item

$d_{\min}=(r_{j^*}-l_{i^*})/(j^*-i^*)$ with $1\leq i^*\leq j^*\leq i-1$.

%

\item
$k_t=i-1$.

\item
$p_{k_s}=l_{k_s}$, i.e., $p_{k_s}$ is at the left endpoint of the interval $I_{k_s}$.
\item
The locations of all points of $P(1,k_s)$ have been explicitly
computed and {\em finalized} (i.e., they will never be changed in the later
algorithm).
\item
For each $1\leq j\leq k_s$, $p_j$ is in $I_j$.
\item
The distance of every pair of adjacent points of $P(1,k_s)$ is at least $d_{\min}$.

\item
For each $j$ with $k_s+1\leq j\leq i-1$, $p_j$ is ``implicitly'' set to
$l_{k_s}+d_{\min}\cdot (j-k_s)$ and $p_j\in I_j$. In other words, the
distance of every pair of adjacent points of $P(k_s,i-1)$ is exactly $d_{\min}$.

%
%
\item
The critical list $\scrL$ has the following {\em priority property}: If $\scrL$ has
more than one element (i.e., $s<t$), then
for any $h$ with $s\leq h\leq t-1$, Inequality \eqref{equ:10} holds for any
$j$ with $k_h+1\leq j\leq i-1$ and $j\neq k_{h+1}$.
\begin{equation}\label{equ:10}
\frac{l_{k_{h+1}}-l_{k_{h}}}{k_{h+1}-k_{h}} >
\frac{l_j-l_{k_{h}}}{j-k_{h}}.
\end{equation}

We give some intuition on what the priority property implies.
Suppose we move all points in $P(k_s+1,{i-1})$ leftwards simultaneously
such that the distances between all adjacent pairs of points
of $P(k_s,i-1)$ keep the same (by the above eighth invariant, they are
the same before the moving).
Then, Inequality \eqref{equ:10} with $h=s$ implies that
$p_{k_{s+1}}$ is the first point of $P(k_s+1,{i-1})$ that arrives at the left
endpoint of its interval. Once $p_{k_{s+1}}$ arrives at the interval
left endpoint, suppose we continue to move the points of
$P(k_{s+1}+1,{i-1})$ leftwards simultaneously
such that the distances between all adjacent pairs of points
of $P(k_{s+1},i-1)$ are the same. Then,
Inequality \eqref{equ:10} with $h=s+1$ makes sure that
$p_{k_{s+2}}$ is the first point of $P(k_{s+1}+1,{i-1})$ that arrives at the left
endpoint of its interval. Continuing the above can explain the
inequality for $h=s+2,s+3,\ldots,t-1$.

%
%

The priority property further leads to the following observation.

\begin{observation}\label{obser:10}
For any $h$ with $s\leq h\leq t-2$, the following holds:
$$\frac{l_{k_{h+1}}-l_{k_{h}}}{k_{h+1}-k_{h}}>
\frac{l_{k_{h+2}}-l_{k_{h+1}}}{k_{h+2}-k_{h+1}}.$$
\end{observation}
\begin{proof}
Note that $k_h+1\leq k_{h+1}< k_{h+2}\leq i-1$.
Let $j=k_{h+2}$.
By Inequality \eqref{equ:10}, we have
\begin{equation}\label{equ:20}
\frac{l_{k_{h+1}}-l_{k_{h}}}{k_{h+1}-k_{h}} >
\frac{l_{k_{h+2}}-l_{k_{h}}}{k_{h+2}-k_{h}}.
\end{equation}

Note that for any four positive numbers $a,b,c,d$ such that $a<c$, $b<d$, and
$\frac{a}{b}>\frac{c}{d}$, it holds that $\frac{a}{b}>\frac{c-a}{d-b}$.
Applying this to Inequality \eqref{equ:20} will obtain the
observation.\qed
\end{proof}
\end{enumerate}

\paragraph{Remark.} By Corollary \ref{coro:10}, Invariants (2), (6),
(7), and (8) together imply that $d_{\min}$ is the optimal objective value for
the sub-problem on the first $i-1$ intervals.

One may verify that initially after $I_1$ is processed,
all invariants trivially hold (we finalize $p_1$ at $l_1$).
In the following we describe the general step of our algorithm to
process the interval $I_i$. We will also show that  all algorithm invariants hold
after $I_i$ is processed.

Depending on whether $p_{i-1}+d_{\min}$ is to the left of
$I_i$, in $I_i$, or to the right of $I_i$, there are three cases.

\subsubsection{The case $p_{i-1}+d_{\min}\leq l_i$}

In this case,
$p_{i-1}+d_{\min}$ is to the left of $I_i$. We set $p_i=l_i$
and finalize it. We do not change $d_{\min}$,
$i^*$, or $j^*$. Further, for each $j\in [k_s+1,i-1]$, we explicitly
compute $p_j=l_{k_s}+d_{\min}\cdot (j-k_s)$ and finalize it.
Finally, we reset $\scrL=\{i\}$.

\begin{lemma}\label{lem:20}
In the case $p_{i-1}+d_{\min}\leq l_i$,
all algorithm invariants hold after $I_i$ is processed.
\end{lemma}
\begin{proof}
Recall that $\scrL=\{i\}$ after $I_i$ is processed. Hence, $k_s=k_t=i$. For the sake
of differentiation, we use $\scrL'=\{k_s',k_{s+1}',\ldots,k'_{t'}\}$
to denote the critical list before we process $I_i$.

\begin{enumerate}

\item

Since $p_i$ is known, Invariant (1) hold.

\item

For Invariant (2), since the same invariant holds before we process $I_i$ and
none of $d_{\min}$, $i^*$, and $j^*$ is changed when we process $I_i$,
Invariant (2) trivially holds after we process $I_i$.

\item
Since $k_t=i$, the third invariant holds.

\item
Recall that $p_{k_s}=p_i=l_i$, which is the fourth invariant.

\item
To prove Invariant (5), since the same invariant holds before $I_i$ is processed,
it is sufficient to show that the points of $P(k_s'+1,i)$ have
been explicitly computed and finalized in the step of processing
$I_i$, which is clearly true according to our algorithm.

\item
To prove Invariant (6), since the same invariant holds  before $I_i$ is processed,
it is sufficient to show that each point
$p_j$ of $P(k_s'+1,i)$ is in $I_j$.

Indeed, consider any $j\in [k_s'+1,i]$.
If $j=i$, then since $p_j=l_j$, it is true that $p_j$ is in $I_j$.
If $j<i$, then by Invariant (8) of $\scrL'$,
$l_{k_s'}+d_{\min}\cdot (j-k_s')$ is in $I_j$. According to our algorithm,
in the step of processing $I_i$, $p_j$ is explicitly set to
$l_{k_s'}+d_{\min}\cdot (j-k_s')$. Hence, $p_j$ is in $I_j$.

\item
To prove Invariant (7), since the same invariant holds  before $I_i$ is processed,
it is sufficient to show that $|p_{i-1}p_i|\geq d_{\min}$, which
is clearly true according to our algorithm.

\item
Invariant (8) trivially holds since $k_s+1>i$ (i.e., there is no
$j$ such that $k_s+1\leq j\leq i$).

\item
Invariant (9) also holds since $\scrL$ has only one element.
\end{enumerate}

This proves that all algorithm invariants hold. The lemma thus
follows. \qed
\end{proof}

\subsubsection{The case $l_i<p_{i-1}+d_{\min}\leq r_i$}

In this case, $p_{i-1}+d_{\min}$ is in $I_i$. We set $p_i=p_{i-1}+d_{\min}$.
We do not change $d_{\min}$, $i^*$, or $j^*$.
We update the critical list $\scrL$ by the following {\em rear-processing
procedure} (because the elements of $\scrL$ are considered from the rear to the
front).

If $s=t$, i.e., $\scrL$ only has one element, then we simply add $i$ to the rear of $\scrL$.
Otherwise, we first check whether the following inequality is true.
\begin{equation}\label{equ:30}
\frac{l_{k_{t}}-l_{k_{t-1}}}{k_{t}-k_{t-1}} > \frac{l_{i}-l_{k_{t-1}}}{i-k_{t-1}}.
\end{equation}

If it is true, then we add $i$ to the end of $\scrL$.

If it is not true, then we remove $k_t$ from $\scrL$ and decrease $t$ by $1$. Next,
we continue to check whether Inequality \eqref{equ:30} (with the
decreased $t$) is true and follow the same procedure until
either the inequality becomes true or $s=t$. In either case, we add $i$ to the end of $\scrL$.
Finally, we increase $t$ by $1$ to let $k_t$ refer to $i$.

This finishes the rear-processing procedure for updating $\scrL$.


\begin{lemma}\label{lem:30}
In the case $l_i< p_{i-1}+d_{\min}\leq r_i$,
all algorithm invariants hold after $I_i$ is processed.
\end{lemma}
\begin{proof}
For the sake of differentiation, we use $\scrL'=\{k_s',k_{s+1}',\ldots,k'_{t'}\}$
to denote the critical list before we process $I_i$. After $I_i$ is processed, we
have $\scrL=\{k_s,k_{s+1},\ldots,k_t\}$.
According to our algorithm, $\scrL$ is obtained from $\scrL'$ by possibly removing
some elements of $\scrL'$ from the rear and then adding $i$ to the end.
Hence, $k_h=k'_h$ for any $h\in [s, t-1]$ and $k_t=i$.
In particular, $k_s=k_s'$ since $\scrL$ has at least two elements (i.e., $s<t$).

\begin{enumerate}

\item
Since the ``temporary'' location of $p_i$ is computed,
the first invariant holds.

\item
The second invariant trivially holds since
none of $d_{\min}$, $i^*$, and $j^*$ is changed when we process $I_i$.

\item

Since $k_t=i$, Invariant (3) holds.

\item

To prove Invariant (4), we need to show that $p_{k_s}=l_{k_s}$. Since
the same invariant holds for $\scrL'$, $p_{k'_s}=l_{k'_s}$. Due to
$k_s=k_s'$, we obtain $p_{k_s}=l_{k_s}$.

\item
Invariant (5) trivially holds since $k_s=k_s'$  and the same invariant holds before $I_i$
is processed.

\item

Similarly, since $k_s=k_s'$, Invariant (6) holds.

\item

Similarly, since $k_s=k_s'$, Invariant (7) holds.

\item

To prove Invariant (8), we need to show that $p_j$ is implicitly set to
$l_{k_s}+d_{\min}\cdot (j-k_s)$ and $p_j\in I_j$ for each $j\in [k_s+1,i]$.

Recall that $k_s=k_s'$ and $d_{\min}$ does not change when we process $I_i$.
Since the same invariant holds before $I_j$ is processed, for $j\in
[k_s+1,i-1]$, it is true that $p_j$ is implicitly set to
$l_{k_s}+d_{\min}\cdot (j-k_s)$ and $p_j\in I_j$. For $j=i$, since
$p_{i}=p_{i-1}+d_{\min}$ and $p_i\in I_i$, $p_i=l_{k_s}+d_{\min}\cdot (i-k_s)$.

Hence, this invariant also holds.
\end{enumerate}

The above has proved that the first eight invariants hold. It remains to prove
the last invariant, i.e., the priority property of $\scrL$.
Our goal is to show that for any $h\in [s,t-1]$,
Inequality \eqref{equ:10} holds for any $j\in [k_h+1,i]$ with $j\neq k_{h+1}$.

Consider any $h\in [s,t-1]$ and any $j\in [k_h+1,i]$
with $j\neq k_{h+1}$. Since $h\leq t-1$, $k'_{h}=k_h$.
Depending on whether $h\leq t-2$ or $h=t-1$, there are two cases.

\paragraph{The case $h\leq t-2$.}
In this case, $h+1\leq t-1$ and thus $k'_{h+1}=k_{h+1}$.

If $j\leq i-1$, then $j\in [k_h+1,i-1]=[k'_h+1,i-1]$.  Since the
priority property holds for $\scrL'$, we have
$\frac{l_{k'_{h+1}}-l_{k'_{h}}}{k'_{h+1}-k'_{h}} >
\frac{l_j-l_{k'_{h}}}{j-k'_{h}}$. As $k'_h=k_h$ and
$k'_{h+1}=k_{h+1}$, Inequality \eqref{equ:10} hold for $j$ and $h$.

If $j=i$, then Inequality \eqref{equ:10} can be proved with the help
of Observation \ref{obser:10}, as follows.

Since $h\leq t-2$ and $s\leq h<t-1$, $k_s$ is not $k_{t-1}$.
Since $k_{t-1}$ is not removed from $\scrL$, according to our algorithm,
Inequality \eqref{equ:30} must be true with replacing $t$ by $t-1$, i.e.,
$\frac{l_{k_{t-1}}-l_{k_{t-2}}}{k_{t-1}-k_{t-2}} > \frac{l_{i}-l_{k_{t-2}}}{i-k_{t-2}}$.

Further, recall that $k_m=k'_m$ for all $m\in [s,t-1]$.
Due to the priority property of
$\scrL'$ and by Observation~\ref{obser:10}, we obtain
$\frac{l_{k_{h+1}}-l_{k_{h}}}{k_{h+1}-k_{h}} >
\frac{l_{k_{t-1}}-l_{k_{t-2}}}{k_{t-1}-k_{t-2}}$.

Combining the above two inequalities gives us
\begin{equation}\label{equ:40}
\frac{l_{k_{h+1}}-l_{k_{h}}}{k_{h+1}-k_{h}} >
\frac{l_{i}-l_{k_{t-2}}}{i-k_{t-2}}.
\end{equation}

Depending on whether $h<t-2$, there are further two subcases.

\begin{enumerate}
\item
If $h=t-2$, then Inequality~\eqref{equ:40} is Inequality \eqref{equ:10}
for $j=i$. So we are done with the proof.

\item
If $h<t-2$, then, $k_h< k_{t-2}\leq i-1$. Recall that $k'_h=k_h$ and $k'_{t-2}=k_{t-2}$.
Due to the priority property of $\scrL'$ and by setting $j=k'_{t-2}$ in Inequality~\eqref{equ:10}, we obtain $\frac{l_{k'_{h+1}}-l_{k'_{h}}}{k'_{h+1}-k'_{h}} > \frac{l_{k'_{t-2}}-l_{k'_{h}}}{k'_{t-2}-k'_{h}}$.

Again, because $k'_h=k_h$, $k'_{h+1}=k_{h+1}$, and $k'_{t-2}=k_{t-2}$, we have
\begin{equation}\label{equ:50}
\frac{l_{k_{h+1}}-l_{k_{h}}}{k_{h+1}-k_{h}} > \frac{l_{k_{t-2}}-l_{k_{h}}}{k_{t-2}-k_{h}}.
\end{equation}

Note that for any positive numbers $x,a,b,c,d$ such that $x>\frac{a}{b}$ and $x>\frac{c}{d}$, it always holds that $x>\frac{a+c}{b+d}$. Applying this to Inequalities \eqref{equ:40} and \eqref{equ:50} leads to $\frac{l_{k_{h+1}}-l_{k_{h}}}{k_{h+1}-k_{h}} > \frac{l_{i}-l_{k_{h}}}{i-k_{h}}$, which is Inequality \eqref{equ:10} for $j=i$.

\end{enumerate}

This proves Inequality \eqref{equ:10} for the case $h\leq t-2$.

\paragraph{The case $h= t-1$.}
In this case, $k_{h+1}=k_t=i$. Due to $j\neq k_{h+1}$, $j\neq i$.

If none of the elements of $\scrL'$ was removed when we updated $\scrL$, i.e.,
$\scrL=\scrL'\cup\{i\}$, then $k_{t-1}=k'_{t'}$.
Since $k'_{t'}=i-1$, $k_h=k_{t-1}=k'_{t'}=i-1$.
Therefore, $k_h+1=i$, and there is no $j$ with $k_h+1\leq j\leq i$ and $j\neq k_{h+1}$ ($=k_t=i$).
Hence, we have nothing to prove for Inequality \eqref{equ:10} in this case.

In the following, we assume at least one element was removed from $\scrL'$ when we
updated $\scrL$. Since $k'_{t-1}=k_{t-1}$ is the last element of $\scrL'$
remaining in $\scrL$, $k'_t$ is the last element removed from $\scrL'$ when we
process $I_i$.  According to the algorithm, $k'_t$
was removed because Inequality~\eqref{equ:30} was not true, i.e., the following holds
\begin{equation}\label{equ:60}
\frac{l_{k'_{t}}-l_{k'_{t-1}}}{k'_{t}-k'_{t-1}} \leq
\frac{l_{i}-l_{k'_{t-1}}}{i-k'_{t-1}}.
\end{equation}

Recall that $k_h+1\leq j\leq i$, $j\neq i$, and
$k_h=k_{t-1}=k'_{t-1}$. Due to the priority property of $\scrL'$ and by
setting $h=t-1$ in Inequality~\eqref{equ:10}, we obtain
\begin{equation}\label{equ:70}
\frac{l_{k'_{t}}-l_{k'_{t-1}}}{k'_{t}-k'_{t-1}} >
\frac{l_j-l_{k'_{t-1}}}{j-k'_{t-1}}.
\end{equation}

Combining Inequalities \eqref{equ:60} and \eqref{equ:70}, we obtain
$\frac{l_{i}-l_{k'_{t-1}}}{i-k'_{t-1}}>\frac{l_j-l_{k'_{t-1}}}{j-k'_{t-1}}$,
which is Inequality \eqref{equ:10} for $h$ and $j$ since $h=t-1$,
$k'_t=k_t=i$, and $k'_{t-1}=k_{t-1}$.

The above proves that the priority property holds for the updated list $\scrL$.

This proves that all algorithm invariants hold after $I_i$ is
processed.  \qed
\end{proof}

\subsubsection{The case $p_{i-1}+d_{\min}>r_i$}

In this case,
$p_{i-1}+d_{\min}$ is to the right of $I_i$.
We first set $p_i=r_i$. Then we perform the following
{\em front-processing procedure}
(because it processes the elements of $\scrL$ from the front to the rear).

If $\scrL$ has only one element (i.e., $s=t$), then we stop.

Otherwise,
we check whether the following is true
\begin{equation}\label{equ:80}
\frac{l_{k_{s+1}}-l_{k_s}}{k_{s+1}-k_s}> \frac{r_i-l_{k_s}}{i-k_s}.
\end{equation}

If it is true, then we perform the following
{\em finalization step}: for each $j=k_s+1,k_s+2,\ldots, k_{s+1}$, we
explicitly compute
$p_{j}=l_{k_s}+\frac{l_{k_{s+1}}-l_{k_s}}{k_{s+1}-k_s}\cdot (j-k_s)$ and
finalize it. Further, we remove $k_s$ from $\scrL$ and increase $s$ by $1$.
Next, we continue the same procedure as above (with the increased
$s$), i.e., first check whether $s=t$, and
if not, check whether Inequality \eqref{equ:80} is true. The
front-processing procedure stops
if either $s=t$ (i.e., $\scrL$ only has one element) or Inequality \eqref{equ:80} is
not true.

After the front-processing procedure,
we update $d_{\min}=(r_i-l_{k_s})/(i-k_s)$, $i^*=k_s$, and $j^*=i$.
Finally, we update the critical list $\scrL$ using the rear-processing procedure, in the same way as in the above second case where $l_i < p_{i-1}+d_{\min}\leq r_i$.
We also ``implicitly'' set $p_j=l_{k_s}+d_{\min}\cdot (j-k_s)$ for each $j\in [k_s+1,i]$ (this is only for the analysis and our algorithm does not do so explicitly).

This finishes the processing of $I_i$.

\begin{lemma}
In the case $p_{i-1}+d_{\min}>r_i$, all algorithm invariants hold
after $I_i$ is processed.
\end{lemma}
\begin{proof}
Let $\scrL=\{k_s,k_{s+1},\ldots,k_t\}$ be the critical list after $I_i$
is processed.  For the sake
of differentiation, we use $\scrL'=\{k_s',k_{s+1}',\ldots,k'_{t'}\}$
to denote the critical list before we process $I_i$.

According to our algorithm,
$\scrL$ is obtained from $\scrL'$ by the following two main steps:
(1) the front-processing step that possibly removes
some elements of $\scrL'$ from the front;
(2) the rear-processing step that possibly removes some elements of
$\scrL'$ from the rear and then adds $i$ to the rear.
Hence, $k_t=i$.

Let $w$ be the index of $\scrL'$ such that $k_s=k'_{w}$.
If $w\neq s$, then $k'_{s},k'_{s+1},\ldots,k'_{w-1}$ are not in $\scrL$.

\paragraph{The first invariant.}
Since the ``temporary'' location of $p_i$ is computed with $p_i=r_i$,
the first invariant holds.

\paragraph{The second invariant.}
By our way of updating $d_{\min}$, $i^*$,
and $j^*$, it holds that $d_{\min}=(r_{j^*}-l_{i^*})/(j^*-i^*)$, with
$1\leq i^*\leq j^*\leq i$. Hence, the invariant holds.

\paragraph{The third invariant.}
Since $k_t=i$, the third invariant trivially holds.

\paragraph{The fourth invariant.}
We need to show that $p_{k_s}=l_{k_s}$.

If $s=w$, then $k_s=k'_s$ and $k_s$ is also
the first element of $\scrL'$. Since the fourth invariant holds before
$I_i$ is processed, $p_{k'_s}=l_{k'_s}$. Thus, we obtain $p_{k_s}=l_{k_s}$.

If $s\neq w$, then when $k'_{w-1}$ was removed from $\scrL$ in the algorithm, the
finalization step explicitly computed
$p_{j}=l_{k'_{w-1}}+\frac{l_{k'_w}-l_{k'_{w-1}}}{k'_{w}-k'_{w-1}}\cdot (j-k'_{w-1})$
for each $j\in [k'_{w-1}+1,k'_w]$. Once can verify that
$p_{k'_w}=l_{k'_w}$. Since $k'_w=k_s$, we obtain $p_{k_s}=l_{k_s}$.

This proves that the fourth invariant also holds.

\paragraph{The fifth invariant.}
Our goal is to show that all points in $P(1,k_s)$ have
been finalized. Since all points in
$P(1,k_s')$ have been finalized before we process $I_i$, it is
sufficient to show that the points for $P(k_s'+1,k_s)$ were finalized
in the step of processing $I_i$.

If $w=s$, then $k_s=k_s'$ and we are done with the proof.
Otherwise, for each $h\in [s,w-1]$, when $k'_h$ was removed from
$\scrL$, the finalization step finalized the points in
$P(k'_h+1,k'_{h+1})$. Hence, all points of $P(k_s'+1,k'_w)$
($=P(k_s'+1,k_s)$) were finalized. Hence, the fifth invariant holds.

\paragraph{The sixth invariant.}
Our goal is to show that for any $p_j$ with $j\in [1,k_s]$, $p_j$ is in $I_j$.

Note that the position of $p_j$ is not changed for any $j\leq k'_s$ when
we process the interval $I_i$.
Since the same invariant holds before we process $I_i$,
$p_j$ is in $I_j$ for any $j\in [1,k_s']$.
Hence, if $k_s=k_s'$, we are done with proof. Otherwise, it
is sufficient to show that $p_j$ is in $I_j$ for any $j\in
[k_{s'}+1,k_s]$.

For $j=k_s$, since $p_j=l_j$, it is trivially true that $p_j$
is in $I_j$. In the following, we assume $j\in [k'_h,k'_{h+1})$ for some
$h\in [s,w-1]$ (recall that $k_s=k'_w$).

According to our algorithm,
$p_j=l_{k'_h}+\frac{l_{k'_{h+1}}-l_{k'_h}}{k'_{h+1}-k'_h}\cdot (j-k_h')$.
Let $d'_{\min}$ be the value of $d_{\min}$ before $I_i$ is processed.
Let $p_j'$ be the original ``temporary'' location of $p_j$ before
$I_i$ is processed.
Since the eighth invariant holds before $I_i$ is processed, we have
$p_j'=l_{k_s'}+d'_{min}\cdot (j-k_s')$ and $p_j'\in I_j$.

We first show that  $p_j\leq p_j'$, i.e., comparing with its original location,
$p_j$ has been moved leftwards in the step of processing $I_i$.
This can be easily seen from the intuitive understanding of the
algorithm. We provide a formal proof below.

Since Invariant (8) holds before $I_i$ is processed,
$p_{k'_{s+1}}$ was implicitly set to
$l_{k'_s}+d'_{\min}\cdot (k'_{s+1}-k'_s)$, which is in $I_{k'_{s+1}}$. Hence,
$l_{k'_s}+d'_{\min}\cdot (k'_{s+1}-k'_s)\geq l_{k'_{s+1}}$. Thus,
$d'_{\min}\geq \frac{l_{k'_{s+1}}-l_{k'_s}}{k'_{s+1}-k'_s}$.
Consequently, $p_j'=l_{k_s'}+d'_{min}\cdot (j-k_s')
\geq l_{k_s'}+\frac{l_{k'_{s+1}}-l_{k'_s}}{k'_{s+1}-k'_s}\cdot (j-k_s')$.

Since the priority property holds for $\scrL'$, by Observation~\ref{obser:10},
$\frac{k'_{s+1}-k'_s}{l_{k'_{s+1}}-l_{k'_s}} \geq
\frac{l_{k'_{h+1}}-l_{k'_{h}}}{k'_{h+1}-k'_{h}}$.
Hence, $p_j=l_{k'_h}+\frac{l_{k'_{h+1}}-l_{k'_h}}{k'_{h+1}-k'_h}\cdot
(j-k_h')\leq l_{k'_h}+\frac{k'_{s+1}-k'_s}{l_{k'_{s+1}}-l_{k'_s}}\cdot (j-k_h')$.

Now to prove $p_j\leq p_j'$, it is sufficient to prove
$\frac{k'_{s+1}-k'_s}{l_{k'_{s+1}}-l_{k'_s}}\geq
\frac{l_{k'_h}-l_{k'_s}}{k'_h-k'_s}$, which is true by Inequality
\eqref{equ:10} (replacing $h$ and $j$ in Inequality \eqref{equ:10} by
$s$ and $k'_h$, respectively) due to the priority property of $\scrL'$.

The above proves that $p_j\leq p_j'$. Since $p_j'\in I_j$, $p_j'\leq r_j$, and
thus, $p_j\leq r_j$. To prove $p_j\in I_j$, it remains to prove $p_j\geq
l_j$.

If $j=k'_h$, then $p_j=l_j$ and we are done with the proof. Otherwise,
due to the priority property of $\scrL'$ and by applying Inequality \eqref{equ:10}, we have
$\frac{l_{k'_{h+1}}-l_{k'_{h}}}{k'_{h+1}-k'_{h}} >
\frac{l_j-l_{k'_{h}}}{j-k'_{h}}$.
Therefore, $p_j=l_{k'_h}+\frac{l_{k'_{h+1}}-l_{k'_h}}{k'_{h+1}-k'_h}\cdot
(j-k_h')> l_j$.

This proves that $p_j$ is in $I_j$. Thus, the sixth invariant holds.

\paragraph{The seventh invariant.}
The goal is to show that the distance of any pair of adjacent points of $P(1,k_s)$ is at least
$d_{\min}$.

Let $d'_{\min}$ be the value of $d_{\min}$ before we process
$I_i$.  We first prove $d'_{\min}> d_{\min}$.

Indeed, if $k_s=k'_s$, then
since the eighth invariant holds before $I_i$ is processed,
$d'_{\min}=\frac{p'_{i-1}-l_{k_s}}{i-1-k_s}$, where $p'_{i-1}$ is the
location of $p_{i-1}$ before we process $I_{i}$.
Recall that $p_{i-1}'+d'_{\min}>r_i$. Hence, we have
$d'_{\min}>\frac{r_i-d'_{\min}-l_{k_s}}{i-1-k_s}$. We can further deduce
$d'_{\min}>\frac{r_i-l_{k_s}}{i-k_s}$. Since
$d_{\min}=\frac{r_i-l_{k_s}}{i-k_s}$, we obtain $d'_{\min}> d_{\min}$.

If $k_s\neq k'_{s}$, since $k'_{w-1}$ was removed from $\scrL$,
Inequality \eqref{equ:80} must hold for $s=w-1$, i.e.,
$\frac{l_{k'_{w}}-l_{k'_{w-1}}}{k'_{w}-k'_{w-1}}>
\frac{r_i-l_{k'_{w-1}}}{i-k'_{w-1}}$. Note that for any four positive
numbers $a,b,c,d$ with $\frac{a}{b}>\frac{c}{d}$, $a<c$, and $b<d$, it
always holds that $\frac{a}{b}>\frac{c-a}{d-b}$. Applying this to the
above inequality gives us $\frac{l_{k'_{w}}-l_{k'_{w-1}}}{k'_{w}-k'_{w-1}}>
\frac{r_i-l_{k'_{w}}}{i-k'_{w}}$.

Since $d_{\min}=\frac{r_i-l_{k_s}}{i-k_s}$ and $k_s=k'_w$, we obtain
$\frac{l_{k'_{w}}-l_{k'_{w-1}}}{k'_{w}-k'_{w-1}}> d_{\min}$.

On the other hand, before $I_i$ is processed, according to the eighth invariant,
$l_{k'_s}+d'_{\min}\cdot (k'_{s+1}-k'_s)$ is in $I_{k'_{s+1}}$. Hence,
$l_{k'_s}+d'_{\min}\cdot (k'_{s+1}-k'_s)\geq l_{k'_{s+1}}$ and
$d'_{\min}\geq\frac{l_{k'_{s+1}}-l_{k'_s}}{k'_{s+1}-k'_s}$.

Further, due to the priority property of $\scrL'$ and by Observation
\ref{obser:10}, it holds that $\frac{k'_{s+1}-k'_s}{l_{k'_{s+1}}-l_{k'_s}}\geq
\frac{l_{k'_{w}}-l_{k'_{w-1}}}{k'_{w}-k'_{w-1}}$.

Combining our above discussions, we obtain $d'_{\min}>d_{\min}$.

Next, we proceed to prove Invariant (7).

Since Invariant (7) holds before $I_i$ is processed, the distance of every pair of
adjacent points of $P(1,k_s')$ is at least $d'_{\min}$.  To prove that
the distance of every pair of adjacent points of $P(1,k_s)$ is at
least $d_{\min}$, since $d'_{\min}>d_{\min}$, if $k_s=k'_s$, then we are
done with the proof, otherwise it is sufficient to show that the
distance of every pair of adjacent points of $P(k_s',k_s)$ is at least
$d_{\min}$.

Consider any $h\in [s,w-1]$.
When $k'_{h}$ is removed from $\scrL$, according to the finalization step,
every pair of adjacent points of $P(k_h',k'_{h+1})$ is
$\frac{l_{k'_{h+1}}-l_{k_{h'}}}{k'_{h+1}-k'_h}$.
Due to the priority property of $\scrL'$ and by Observation \ref{obser:10},
$\frac{l_{k'_{h+1}}-l_{k'_h}}{k'_{h+1}-k'_h}\geq
\frac{l_{k'_{w}}-l_{k'_{w-1}}}{k'_{w}-k'_{w-1}}$.
Recall that we have proved above that
$\frac{l_{k'_{w}}-l_{k'_{w-1}}}{k'_{w}-k'_{w-1}}>d_{\min}$.
Hence, we obtain that the distance of every pair of adjacent points of
$P(k_h',k'_{h+1})$ is at least $d_{\min}$.
This further implies that the distance of every pair of adjacent points of
$P(k_s',k'_{w})$ ($=P(k_s',k_s)$) is at least $d_{\min}$.

Hence, the seventh invariant holds.

\paragraph{The eighth invariant.}
Consider any $j\in [k_s,i]$. Based on our algorithm,  $p_j$ is implicitly
set to $l_{k_s}+d_{\min}\cdot (j-k_s)$.  Hence, to prove the invariant, it
remains to show that $p_j$ is in $I_j$.

If $j=i$, then since $p_i=r_i$, it is true that $p_j\in
I_j$. In the following, we assume $j\leq i-1$.

Let $p_j'$ be the ``temporary'' location of $p_j$ before $I_i$ is processed.
Since the eighth invariant holds before $I_i$ is processed,
$p_j'=l_{k_s'}+d'_{\min}\cdot (j-k_s')$ and $p_j'\in I_j$.
Again, let $d'_{\min}$ be the value of $d_{\min}$ before we process
$I_i$.  Recall that we have proved above that $d'_{\min}> d_{\min}$.

We claim that $p_j\leq p_j'$. Indeed, if $k_s=k_s'$, then $p_j\leq p_j'$ follows
from $d'_{\min}> d_{\min}$.  Otherwise, note that
$p_j'=l_{k_s'}+d'_{\min}\cdot (k_w'-k_s') + d'_{\min}\cdot (j-k_w')=
p'_{k'_w}+d'_{\min}\cdot (j-k_w')$, where $p'_{k'_w}$ is the
``temporary'' location of $p_{k'_w}$ before $I_i$ is processed.
Since $k'_w=k_s$, we have $p_j'=p'_{k_s}+d'_{\min}\cdot (j-k_s)$.

Since Invariant (8) holds before $I_i$ is processed, $p'_{k_s}$ is in $I_{k_s}$. Hence,
$p'_{k_s}\geq l_{k_s}$. Therefore, we obtain
$p_j'\geq l_{k_s}+d'_{\min}\cdot (j-k_s)\geq l_{k_s}+d_{\min}\cdot (j-k_s)=p_j$.

This proves the above claim that $p_j\leq p_j'$.

Since $p_j'\in I_j$ and $p_j\leq p_j'$, we obtain $p_j\leq r_j$. To prove
$p_j\in I_j$, it remains to show $p_j\geq l_j$, as follows.

According to our algorithm, $k_s$
was not removed from $\scrL$ either because $k_s$ is the last element
of $\scrL'$ or because Inequality \eqref{equ:80} is not true.

In the former case, it holds that $k_s=i-1$. Since $j\in [k_s,i-1]$,
$j=k_s$. Due to $p_{k_s}=l_{k_s}$, we obtain $p_j\geq l_j$.

In the latter case, $k_s$ is not the last element of $\scrL'$ that is in $\scrL$.
Since $k'_w=k_s$, we have $k'_{w+1}=k_{s+1}$.
Due to the priority property of $\scrL'$ and by
Inequality \eqref{equ:10} (with $h=w$), we have
$\frac{l_{k'_{w+1}}-l_{k'_{w}}}{k'_{w+1}-k'_{w}} \geq
\frac{l_j-l_{k'_{w}}}{j-k'_{w}}$.
Since $k_s=k'_{w}$ and $k_{s+1}=k'_{w+1}$, it holds that
$\frac{l_{k_{s+1}}-l_{k_{s}}}{k_{s+1}-k_{s}} \geq
\frac{l_j-l_{k_{s}}}{j-k_{s}}$.
Since Inequality \eqref{equ:80} is not true, we further obtain
$\frac{r_i-l_{k_{s}}}{i-k_{s}} \geq \frac{l_j-l_{k_{s}}}{j-k_{s}}$.
Recall that $d_{\min}=\frac{r_i-l_{k_s}}{i-k_s}$. Hence, $d_{\min}\geq
\frac{l_j-l_{k_{s}}}{j-k_{s}}$ and
$p_j=l_{k_s}+d_{\min}\cdot (j-k_s)\geq l_j$.

This proves that the eighth invariant holds.

\paragraph{The ninth invariant.}
Our goal is to prove that the priority property holds for $\scrL$.
Since the priority property holds for $\scrL'$,
intuitively we only need to take care of the ``influence'' of
$i$ (i.e., some elements were possibly removed from the rear of $\scrL'$ and
$i$ was added to the rear in the rear-processing procedure).
Note that although some elements were also possibly removed from the front of
$\scrL'$ in the front-processing procedure, this does not affect the
priority property of the remaining elements of the list. Hence, to prove that the priority property holds for
$\scrL$, we have exactly the same situation as in Lemma~\ref{lem:30}.
Hence, we can use the same proof
as that for Lemma~\ref{lem:30}. We omit the details.

This proves that all algorithm invariants hold after $I_i$ is processed.
The lemma thus follows.
\qed
\end{proof}

The above describes a general step of the algorithm for processing the
interval $I_i$. In addition, if $i=n$ and $k_s<n$, we also need to perform the
following additional finalization step: for each $j\in [k_s+1, n]$, we explicitly compute
$p_j=l_{k_s}+d_{\min}\cdot (j-k_s)$
and finalize it.
This finishes the algorithm.

\subsection{The Correctness and the Time Analysis}
\label{sec:correct}

Based on the algorithm invariants and Corollary~\ref{coro:10}, the following
lemma proves the correctness of the algorithm.

\begin{lemma}\label{lem:50}
The algorithm correctly computes an optimal solution.
\end{lemma}
\begin{proof}
Suppose $P=\{p_1,p_2,\ldots,p_n\}$ is the set of points computed by the
algorithm.
Let $d_{\min}$ be the value and $\scrL=\{k_s,k_{s+1},\ldots, k_t\}$ be the critical list after the algorithm finishes.

We first show that for each $j\in [1,n]$, $p_j$ is in $I_j$.
According to the sixth algorithm invariant of $\scrL$, for each $j\in [1,k_s]$, $p_j$ is in
$I_j$. If $k_s=n$, then we are done with the proof. Otherwise, for each $j\in
[k_s+1,n]$, according to the additional finalization step after $I_n$ is processed,
$p_j=l_{k_s}+d_{\min}\cdot (j-k_s)$, which is in $I_j$ by the eighth algorithm
invariant.

Next we show that the distance of every pair of adjacent points of $P$ is at least $d_{\min}$.
By the seventh algorithm invariant, the distance of
every pair of adjacent points of $P(1,k_s)$ is at least $d_{\min}$. If $k_s=n$,
then we are done with the proof. Otherwise, it is sufficient to show that the distance of
every pair of adjacent points of $P(k_s,n)$ is at least $d_{\min}$,
which is true according to
the additional finalization step after $I_n$ is processed.

The above proves that $P$ is a {\em feasible solution} with respect to
$d_{\min}$, i.e., all points of $P$ are in their corresponding intervals and
the distance of every pair of adjacent points of $P$ is at least $d_{\min}$.

To show that $P$ is also an optimal solution, based on the second algorithm invariant,
it holds that $d_{\min}=\frac{r_{j^*}-l_{i^*}}{j^*-i^*}$. By Corollary \ref{coro:10},
$d_{\min}$ is an optimal objective value. Therefore, $P$ is an optimal solution.
\qed
\end{proof}

The running time of the algorithm is analyzed in the proof of
Theorem \ref{theo:10}. The pseudocode is given in Algorithm \ref{algo:line}.

\begin{theorem}\label{theo:10}
Our algorithm computes an optimal solution of the line version of
points dispersion problem in $O(n)$ time.
\end{theorem}
\begin{proof}
In light of Lemma \ref{lem:50},
we only need to show that the running time of the algorithm is $O(n)$.

To process an interval $I_i$, according to our algorithm, we only spend $O(1)$
time in addition to two possible
procedures: a front-processing procedure and a rear-processing procedure.
Note that the front-processing procedure may contain several finalization steps.
There may also be an additional finalization step after $I_n$ is
processed. For the purpose of analyzing the total running time of the algorithm,
we exclude the finalization steps from the front-processing procedures.

For processing $I_i$,
the front-processing procedure (excluding the time of the finalization steps)
runs in $O(k+1)$ time where $k$ is the
number of elements removed from the front of the critical list $\scrL$. An easy
observation is that any element can be removed from $\scrL$
at most once in the entire
algorithm. Hence, the total time of all front-processing procedures
in the entire algorithm is $O(n)$.

Similarly, for processing $I_i$, the rear-processing procedure runs in $O(k+1)$ time where $k$ is the
number of elements removed from the rear of $\scrL$.
Again, since any element can be removed from $\scrL$ at most once in the entire
algorithm,
the total time of all rear-processing procedures in the entire algorithm is $O(n)$.

Clearly, each point is finalized exactly once in the entire algorithm. Hence,
all finalization steps in the entire algorithm together take $O(n)$ time.

Therefore, the algorithm runs in $O(n)$ time in total.
\qed
\end{proof}

\begin{algorithm}[H]
\caption{The algorithm for the line version of the problem}
\label{algo:line}
\KwIn{$n$ intervals $I_1,I_2,\ldots,I_n$ sorted from left to right on $\ell$ }
\KwOut{$n$ points $p_1,p_2,\ldots,p_n$ with $p_i\in I_i$ for each $1\leq i\leq n$} \BlankLine
$p_1\leftarrow l_1$, $i^*\leftarrow 1$, $j^*\leftarrow 1$, $d_{\min}\leftarrow \infty$,
$\scrL\leftarrow \{1\}$\;
\For{$i\leftarrow 2$ \KwTo $n$}
{
	\eIf{$p_{i-1} + d_{\min}\leq l_i$}
	{
		$p_i \leftarrow l_i$, $\scrL\leftarrow \{i\}$\;
    }
	{
		\eIf{$l_i< p_{i-1} + d_{\min}\leq  r_i$}
		{
		   $p_i\leftarrow p_{i-1} + d_{\min}$\;
	        }
                     (\tcc*[h]{$ p_{i-1} + d_{\min}>  r_i$})
		     {
		       $p_i\leftarrow r_i$,
		       $k_s \leftarrow \text{the front		       element of } \scrL $\;
		     \While(\tcc*[f]{the front-processing procedure}){$|\scrL|>1$}
		     {
		        \eIf{$\frac{l_{k_{s+1}}-l_{k_s}}{k_{s+1}-k_s}> \frac{r_i-l_{k_s}}{i-k_s}$}
			{
                          \For{$j\leftarrow k_s+1$ \KwTo $k_{s+1}$}
			  {
			     $p_j\leftarrow l_{k_s}+\frac{l_{k_{s+1}}-l_{k_s}}{k_{s+1}-k_s}\cdot (j-k_s)$\;
			  }
			  remove $k_s$ from $\scrL$, \ \
			  $k_s \leftarrow \text{the front element of } \scrL $\;
			}
			{
			  break\;
			}
		     }
		     $i^*\leftarrow k_s$, $j^*\leftarrow i$, $d_{\min}\leftarrow
		     \frac{r_{j^*}-l_{i^*}}{j^*-i^*}$\;
		     }
		   \While(\tcc*[f]{the rear-processing procedure}){$|\scrL|>1$}
		   {
		      $k_t \leftarrow \text{the rear element of } \scrL $\;
		      \lIf{$\frac{l_{k_{t}}-l_{k_{t-1}}}{k_{t}-k_{t-1}} >\frac{l_{i}-l_{k_{t-1}}}{i-k_{t-1}}$}
		      {
		        break\;
		      }
		      {
		        remove $k_t$ from $\scrL$\;
		      }
		   }
		   add $i$ to the rear of $\scrL$\;
		}
}
$k_s \leftarrow \text{the front element of } \scrL $\;
\If{$k_s<n$}
{
                \For{$j\leftarrow k_s+1$ \KwTo $n$}
		 {
		  $p_j\leftarrow l_{k_s}+d_{\min}\cdot (j-k_s)$\;
		 }
}
\end{algorithm}

\section{The Cycle Version}
\label{sec:cycle}

In the cycle version, the intervals of $\calI=\{I_1,I_2,\ldots,I_n\}$ in their index order
are sorted cyclically on $\calC$. Recall that the intervals of $\calI$ are
pairwise disjoint.


For each $i\in [1,n]$, let $l_i$ and $r_i$ denote the two endpoints of
$I_i$, respectively, such that if we move from $l_i$ to $r_i$
clockwise on $\calC$, we will always stay on $I_i$.

For any two points $p$ and $q$ on $\calC$, we use $|\rarrow{pq}|$ to denote
the length of the arc of $\calC$ from $p$ to $q$ clockwise, and thus the distance of
$p$ and $q$ on $\calC$ is $\min\{|\rarrow{pq}|,|\rarrow{qp}|\}$.

For each interval $I_i\in \calI$, we use $|I_i|$ to denote its length; note that
$|I_i|=|\rarrow{l_ir_i}|$.  We use $|\calC|$ to denote the total length of $\calC$.

Our goal is to find a point $p_i$ in $I_i$ for each $i\in [1,n]$ such that
the minimum distance between any pair of these points, i.e., $\min_{1\leq
i< j\leq
n}|p_ip_j|$, is maximized.

Let $P=\{p_1,p_2,\ldots,p_n\}$ and let $d_{opt}$ be the optimal
objective value. It is obvious that $d_{opt}\leq \frac{|\calC|}{n}$.
Again, for simplicity of discussion, we make a general position assumption that no two
endpoints of the intervals have the same location on $\calC$.

\subsection{The Algorithm}

The main idea is to convert the problem to a problem instance on a line and then
apply our line version algorithm.
More specifically, we copy all intervals of $\calI$ twice to a line $\ell$
and then apply our line version algorithm on these $2n$ intervals. The line
version algorithm will find $2n$ points in these intervals. We will show that
a subset of $n$ points in $n$ consecutive intervals correspond
to an optimal solution for our original problem on $\calC$. The
details are given below.

Let $\ell$ be the $x$-axis. For each $1\leq i\leq n$, we create an interval
$I_i'=[l_i',r_i']$ on $\ell$ with $l_i'=|\rarrow{l_1l_i}|$ and
$r_i'=l_i'+|I_i|$, which is actually a copy of $I_i$.
In other words, we first put a copy $I_1'$ of $I_1$ at $\ell$ such
that its left endpoint is at $0$ and then we continuously copy other
intervals to $\ell$ in such a way that the pairwise distances of the intervals on $\ell$
are the same as the corresponding
clockwise distances of the intervals of $\calI$ on $\calC$. The above only makes one copy
for each interval of $\calI$. Next, we make another copy for each
interval of $\calI$ in a similar way: for each $1\leq i\leq n$, we create
an interval $I_{i+n}'=[l_{i+n}',r_{i+n}']$ on $\ell$ with $l_{i+n}'=l_i'+|\calC|$
and $r_{i+n}'=r_i'+|\calC|$.
Let $\calI'=\{I_1',I_2',\ldots,I_{2n}'\}$. Note that the intervals of $\calI'$
in their index order are sorted from left to right on $\ell$.


We apply our line version algorithm on the intervals of $\calI'$. However, a subtle change is that here we initially set $d_{\min}=\frac{|\calC|}{n}$ instead of
$d_{\min}=\infty$. The rest of the algorithm is the same as before.
We want to emphasize that this change on initializing $d_{\min}$ is necessary to guarantee the correctness of our algorithm for the cycle version.
A consequence of this change 
is that after the algorithm
finishes, if $d_{\min}$ is still equal to $\frac{|\calC|}{n}$, then $\frac{|\calC|}{n}$ may
not be the optimal objective value for the above line version problem, but if
$d_{\min}<\frac{|\calC|}{n}$, then $d_{\min}$ must be the optimal objective
value. As will be clear later, this does not affect our final solution
for our original problem on the cycle $\calC$.  Let $P'=\{p_1',\ldots,p_{2n}'\}$ be the
points computed by the line version algorithm with $p_i'\in I_i'$ for each $i\in [1,2n]$.

Let $k$ be the largest index in $[1,n]$ such that $p_k'=l_k'$. Note that such an
index $k$ always exists since $p_1'=l_1'$.  Due to that we initialize
$d_{\min}=\frac{|\calC|}{n}$ in our line version algorithm, we can prove the following lemma.

\begin{lemma}\label{lem:60}
It holds that $p_{k+n}'=l_{k+n}'$.
\end{lemma}
\begin{proof}
We prove the lemma by contradiction. Assume to the contrary that
$p_{k+n}'\neq l_{k+n}'$. Since $p_{k+n}'\in I_{k+n}'$, it must be that
$p'_{k+n}>l'_{k+n}$. Let $p_i'$ be the rightmost point of $P'$
to the left of $p'_{k+n}$ such that $p_i'$ is at the left endpoint of its
interval $I_i'$. Depending on whether $i\leq n$, there are two cases.

\begin{enumerate}
\item
If $i> n$, then let $j=i-n$. Since $i<k+n$, $j<k$.
We claim that $|p_{j}'p_k'|<|p_{j+n}'p_{n+k}'|$.

Indeed, since $p_j'\geq l_j'$ and $p_k'=l_k'$, we have $|p_j'p_k'|\leq |l_j'l_k'|$.
Note that $|l_j'l_k'|=|l_{j+n}'l_{k+n}'|$. On the other hand, since
$p_{j+n}'=l_{j+n}'$ and $p_{k+n}'>l_{k+n}'$, it holds that
$|p_{j+n}'p_{k+n}'|>|l_{j+n}'l_{k+n}'|$. Therefore, the claim follows.


Let $d$ be the value of $d_{\min}$ right before the algorithm processes $I_{i}'$.
Since during the execution of our line version algorithm $d_{\min}$ is monotonically decreasing,
it holds that $|p_j'p_k'|\geq d\cdot (k-j)$.
Further, by the definition of $i$, for any $m\in [i+1,k+n]$,
$p_{m}'>l_{m}'$. Thus, according to our line version algorithm,
the distance of every adjacent pair of points of $p_{i}',p_{i+1}'\ldots,p_{k+n}'$
is at most $d$. Thus, $|p_{i}'p_{k+n}'|\leq d\cdot (k+n-i)$.
Since $j=i-n$, we have $|p_{j+n}'p_{k+n}'|\leq d\cdot (k-j)$.
Hence, we obtain $|p_j'p_k'|\geq |p_{j+n}'p_{k+n}'|$. However, this contradicts with our
above claim.

\item
If $i\leq n$, then by the definition of $k$, we have $i=k$.
Let $d$ be the value of $d_{\min}$ right before the algorithm processes $I_i'$.
By the definition of $i$, the distance of every adjacent pair of points of
$p_k',p_{k+1}'\ldots,p_{k+n}'$ is at most $d$.
Hence, $|p_k'p_{k+n}'|\leq n\cdot d$. Since $p_k'=l_k'$ and $p_{n+k}'>l_{n+k}'$, we have
$|p_k'p_{n+k}'|>|l_k'l_{n+k}'|=|\calC|$. Therefore, we obtain that
$n\cdot d>|\calC|$.

However, since we initially set
$d_{\min}=|\calC|/n$ and the value $d_{\min}$ is monotonically
decreasing during the execution of the algorithm, it must hold that $n\cdot
d\leq |\calC|$. We thus obtain contradiction.
\end{enumerate}

Therefore, it must hold that $p'_{n+k}=l'_{n+k}$. The lemma thus follows.
\qed
\end{proof}

We construct a solution set $P$ for our cycle version problem by mapping the points
$p_k',p_{k+1}',\ldots,p_{n+k-1}'$ back to $\calC$. Specifically, for each
$i\in [k,n]$, we put $p_i$ at a point on $\calC$ with a distance $p_i'-l_i'$
clockwise from $l_i$; for each
$i\in [1,k-1]$, we put $p_{i}$ at a point on $\calC$ at a distance
$p_{i+n}'-l_{i+n}'$ clockwise from $l_{i}$.
Clearly,  $p_i$ is in $I_i$ for each $i\in [1,n]$.
Hence, $P$ is a ``feasible'' solution for our cycle version problem.
Below we show that $P$ is actually an optimal solution.

Consider the value $d_{\min}$ returned by the line version algorithm after all
intervals of $\calI'$ are processed. Since the distance of every
pair of adjacent points of $p_k',p_{k+1}',\ldots, p_{n+k}'$ is at
least $d_{\min}$, $p_k'=l_k'$, $p_{n+k}'=l_{n+k}'$ (by
Lemma \ref{lem:60}), and $|l_k'l_{n+k}'|=|\calC|$,
by our way of constructing $P$, the distance of
every pair of adjacent points of $P$ on $\calC$ is at least $d_{\min}$.


Recall that $d_{opt}$ is the optimal object value of our cycle version
problem. The following lemma implies that $P$ is an optimal solution.

\begin{lemma}\label{lem:90}
$d_{\min}= d_{opt}$.
\end{lemma}
\begin{proof}
Since $P$ is a feasible solution with respect to $d_{\min}$, $d_{\min}\leq
d_{opt}$ holds.

If $d_{\min}=|\calC|/n$, since $d_{opt}\leq |\calC|/n$, we obtain $d_{opt}\leq
d_{\min}$. Therefore, $d_{opt}= d_{\min}$, which leads to the lemma.

In the following, we assume $d_{\min}\neq |\calC|/n$.
Hence, $d_{\min}<|\calC|/n$. According to our line
version algorithm, there must
exist $i^*<j^*$ such that $d_{\min}=\frac{r'_{j^*}-l'_{i^*}}{j^*-i^*}$.
We assume there is no $i$ with $i^*<i<j^*$ such that
$d_{\min}=\frac{r'_{j^*}-l'_{i}}{j^*-i}$ since otherwise we could change $i^*$ to
$i$.
Since $d_{\min}=\frac{r'_{j^*}-l'_{i^*}}{j^*-i^*}$, it is necessary that
$p_{i^*}'=l'_{i^*}$ and $p_{j^*}'=r'_{j^*}$. By the above assumption,
there is no $i\in [i^*,j^*]$ such that $p_i'=l'_i$.
Since $p_k'=l'_k$ and $p_{k+n}'=l'_{k+n}$ (by
Lemma \ref{lem:60}), one of the following three cases must be true:
$j^*<k$, $k\leq i^*<j^*<n+k$, or
$n+k\leq i^*$.  In any case, $j^*-i^*<n$.
By our way of defining $r'_{j^*}$ and $l'_{i^*}$, we have the
following:
\begin{equation*}
d_{\min}=\frac{r'_{j^*}-l'_{i^*}}{j^*-i^*}=
\begin{cases}
|\rarrow{l_{i^*}r_{j^*}}|/(j^*-i^*), & \text{if $j^*\leq n$},\\
|\rarrow{l_{i^*}r_{j^*-n}}|/(j^*-i^*), & \text{if $i^*\leq n<j^*$},\\
|\rarrow{l_{i^*-n}r_{j^*-n}}|/(j^*-i^*) & \text{if $n<i^*$}.
\end{cases}
\end{equation*}


We claim that $d_{opt}\leq d_{\min}$ in all three cases: $j^*\leq n$, $i^*\leq
n<j^*$, and $n<i^*$. In the
following we only prove the claim in the first case where $j^*\leq n$ since the other two
cases can be proved analogously (e.g., by re-numbering the indices).

Our goal is to prove $d_{opt}\leq \frac{|\rarrow{l_{i^*}r_{j^*}}|}{j^*-i^*}$.
Consider any optimal solution in which
the solution set is $P=\{p_1,p_2,\ldots, p_n\}$.  Consider the points
$p_{i^*},p_{i^*+1},\ldots,p_{j^*}$, which are in the
intervals $I_{i^*},I_{i^*+1},\ldots,I_{j^*}$. Clearly,
$|\rarrow{p_kp_{k+1}}|\geq d_{opt}$ for any $k\in [i^*,j^*-1]$.
Therefore, we have $|\rarrow{p_{i^*}p_{j^*}}|\geq d_{opt}\cdot
(j^*-i^*)$. Note that $|\rarrow{p_{i^*}p_{j^*}}|\leq
|\rarrow{l_{i^*}r_{j^*}}|$. Consequently, we obtain
$d_{opt}\leq \frac{|\rarrow{l_{i^*}r_{j^*}}|}{j^*-i^*}$.


Since both $d_{\min }\leq d_{opt}$ and $d_{opt}\leq d_{\min}$,
it holds that $d_{opt}=d_{\min}$. The lemma thus follows.
\qed\end{proof}

The above shows that $P$ is an optimal solution with $d_{opt}=d_{\min}$.
The running time of the algorithm is $O(n)$ because the line version algorithm
runs in $O(n)$ time.
As a summary, we have the following theorem.

\begin{theorem}\label{theo:30}
The cycle version of the points dispersion problem is solvable in $O(n)$ time.
\end{theorem}

\section{Concluding Remarks}
\label{sec:conclusion}

In this paper we present a linear time algorithm for the point dispersion problem on
disjoint intervals on a line. Further, by making use of this algorithm, we also
solve the same problem on a cycle in linear time.

It would be interesting to consider the general case of the problem
in which the intervals may overlap. In fact, for the line version, if we know the order of the
intervals in which the sought points in an optimal solution are sorted
from left to right, then we can apply our algorithm to
process the intervals in that order and the obtained solution is an
optimal solution. For example, if no interval is allowed to
contain another completely, then there must exist an optimal solution in which
the sought points from left to right correspond to the intervals
ordered by their left (or right) endpoints. Hence, to solve the
general case of the line version problem, the key is to find an order of
intervals. This is also the case for the cycle version.

\paragraph*{Acknowledgment.}
The authors would like to thank Minghui Jiang for suggesting the problem to them.
This research was supported in part by NSF under Grant CCF-1317143.

\bibliography{reference}
\bibliographystyle{plain}

\end{document}